\documentclass[11pt]{CGC2}

\usepackage[english]{babel}
\usepackage[latin1]{inputenc}
\usepackage[T1]{fontenc}
\usepackage{amsmath}
\usepackage{amsfonts}
\usepackage{amssymb}
\usepackage{verbatim}
\usepackage{rotating}
\usepackage{textcomp}
\usepackage{colortbl}
\usepackage{xcolor}
\usepackage{graphicx}
\usepackage{fancyhdr}
\usepackage{algorithm} 
\usepackage{algorithmic}
\usepackage{amsthm}
\usepackage{indentfirst}
\usepackage[usenames,dvipsnames]{pstricks}
\usepackage{epsfig}

\usepackage[all]{xy}

\newcommand{\A}{{\mathcal A}}
\newcommand{\B}{{\mathcal B}}

\newcommand{\V}{{\mathcal V}}

\newcommand{\NLI}{{\mathcal N}}

\newcommand{\nP}{{\mathfrak n}}

\newcommand{\w}{\mathrm{w}}
\newcommand{\N}{\mathrm{N}}
\newcommand{\dd}{\mathrm{d}}

\newcommand{\op}{{\sf p}} 

\newcommand{\ef}{\underline{f}}
\newcommand{\eg}{\underline{g}}

\newcommand{\dist}{\mathrm{d}}

\begin{document}

\Logo{}

\begin{frontmatter}

\title{Nonlinearity of Boolean functions: an algorithmic approach based on
multivariate polynomials.}
\runtitle{Nonlinearity of Boolean functions and multivariate polynomials.}

{\author{Emanuele Bellini}} 
{\tt{(eemanuele.bellini@gmail.com)}}\\
{Department of Mathematics, University of Trento, Italy.}

{\author{Ilaria Simonetti}}
{{\tt (ilaria.simonetti@gmail.com)}}\\
{{Department of Mathematics, University of Milan, Italy.}}

{\author{Massimiliano Sala}} 
{\tt{(maxsalacodes@gmail.com)}}\\
{Department of Mathematics, University of Trento, Italy.}

\runauthor{E.~Bellini, I.~Simonetti, M.~Sala}

\begin{abstract}
We compute the nonlinearity of Boolean functions with Gröbner basis techniques, providing two algorithms: one over the binary field and the other over the rationals. We also estimate their complexity. Then we show how to improve our rational algorithm, arriving at a worst-case complexity of  $O(n2^n)$ operations over the integers, that is, sums and doublings. This way, with a different approach, we reach the same complexity of established algorithms, such as those based on the fast Walsh transform.
\end{abstract}


\end{frontmatter}

\section{Introduction}
  \label{secIntro}
Any function from $(\FF_2)^n$ to $\FF_2$ is called a Boolean function. Boolean functions are important in symmetric cryptography, since they are used in the confusion layer of ciphers. An affine Boolean function does not provide an effective confusion. To overcome this, we need functions which are as far as possible from being an affine function. The effectiveness of these functions is measured by several parameters, one of these is called ``nonlinearity'' (\cite{CGC-cd-book-carlet}).\\
In this paper, we provide three methods to compute the nonlinearity of Boolean functions. 
Moreover, we give an estimate of the complexity of our methods, comparing it with the complexity of the classical method which uses the fast Walsh transform and the fast M\"obius transform.
\\
In Sections \ref{secPrelOnBF} and \ref{secPrelOnGB} we recall the basic notions and statements, especially regarding Boolean functions, which are necessary for our methods.\\
In Section \ref{secNLwithGBoverF2} and \ref{secNLwithGBoverQ} we provide two algorithms which reduce the problem of computing the nonlinearity of a Boolean function to that of solving a \Gr\ basis. In particular, in Section \ref{secNLwithGBoverQ} we associate to each Boolean function in $n$ variables a polynomial whose evaluations represent the distance from all possible affine functions.\\
In Section \ref{secNLwithFPE} we show that this polynomial can be used to find the nonlinearity of a Boolean function without passing through a \Gr\ basis computation. In Section \ref{secNLPolProp} we provide some results to express the coefficients of this polynomials, and we show in Section \ref{secNLPolComplex} that these can be computed also using fast transforms.\\
Finally, in Section \ref{secNLComplexity} we analyze the complexity of the proposed methods, both experimentally and theoretically. 
In particular, we show that using fast Fourier methods we arrive at a worst-case complexity of  $O(n2^n)$ operations over the integers, that is, sums and doublings. This way, with a different approach, we reach the same complexity of established algorithms, such as those based on the fast Walsh transform.


\section{Preliminaries and Notation on Boolean functions}
  \label{secPrelOnBF}

In this chapter we summarize some definitions and known results from \cite{CGC-cd-book-carlet} and \cite{CGC-cd-book-macwilliamsI}, concerning Boolean functions and the classical techniques to determine their nonlinearity.\\

We denote by $\FF$ the field $\FF_2$. The set $\FF^n$ is the set of all binary vectors of length $n$, viewed as an $\FF$-vector space.\\
Let $v\in\FF^n$. The \emph{Hamming weight} $\w(v)$ of the vector $v$ is the number of its nonzero coordinates. For any two vectors $v_1,v_2\in\FF^n$, the \emph{Hamming distance} between $v_1$ and $v_2$, denoted by $\dist(v_1,v_2)$, is the number of coordinates in which the two vectors differ.\\
A \emph{Boolean function} is a function $f:\FF^n\rightarrow \FF$. The set of all Boolean functions from $\FF^n$ to $\FF$ will be denoted by $\B_n$.


\subsection{Representations of Boolean functions}
%
\subsubsection{Evaluation vector}
We assume implicitly to have ordered $\FF^n$, so that $\FF^n=\{\op_1,\ldots,\op_{2^n}\}$.\\ 
A Boolean function $f$ can be specified by a \emph{truth table}, which gives the evaluation of $f$ at all $\op_i$'s.
%
\begin{definition}
We consider the evaluation map:
$$
\B_n \longrightarrow \FF^{2^n} 
\qquad
f \longmapsto \underline{f}=(f(\op_1),\ldots,f(\op_{2^n}))\,.
$$
The vector $\underline{f}$ is called the \emph{evaluation vector} of $f$.
\end{definition}
%
%
Once the order on $\FF^n$ is chosen, i.e. the $\op_i$'s are fixed, it is clear that the evaluation vector of $f$ uniquely identifies $f$.
\subsubsection{Algebraic normal form}
A Boolean function $f\in\B_n$ can be expressed in a unique way as a square free polynomial in $\FF[X]=\FF[x_1,\ldots,x_n]$, i.e.
$$f=\sum_{v \in \FF^n}b_vX^v\,,$$
where $X^v=x^{v_1}\cdots x^{v_n}$.\\
This representation is called the \emph{Algebraic Normal Form} (ANF).\\
%
\begin{definition}
The degree of the ANF of a Boolean function $f$ is called the \emph{algebraic degree} of f, denoted by $\deg f$, and it is equal to 
$\max\{\w(v) \mid v \in \FF^n, b_v \ne 0 \}$.
\end{definition}
Let $\A_n$ be the set of all affine functions from $\FF^n$ to $\FF$, i.e. the set of all Boolean functions in $\B_n$ with algebraic degree 0 or 1. If $\alpha\in\A_n$ then its ANF can be written as
$$\alpha(X)=a_0 + \sum_{i=1}^na_ix_i\,.$$
%
%
%
\indent
There exists a simple divide-and-conquer butterfly algorithm (\cite{CGC-cd-book-carlet}, p.10) to compute the ANF from the truth-table (or vice-versa) of a Boolean function, which requires $O(n2^{n})$ bit sums, while $O(2^n)$ bits must be stored. This algorithm is known as the \emph{fast M\"obius transform}.
\subsubsection{Numerical normal form}
%
In \cite{CGC-cry-art-carlet1999} a useful representation of Boolean functions for characterizing several cryptographic criteria (see also \cite{CGC-cry-art-carlet2001bent}, \cite{CGC-cry-carlet2002coset}) is introduced.\\
Boolean functions can be represented as elements of $\KK[X]/\langle X^2-X \rangle$, where $\langle X^2-X \rangle$ is the ideal generated by the polynomials $x_1^2-x_1,\ldots,x_n^2-x_n$, and $\KK$ is $\ZZ$, $\QQ$, $\RR$, or $\CC$.
\begin{definition}\label{defNNF}
 Let $f$ be a function on $\FF^n$ taking values in a field $\KK$. We call the \emph{numerical normal form (NNF)} of $f$ the following expression of $f$ as a polynomial:
 $$
 f(x_1,\ldots,x_n) = \sum_{u \in \FF^n}\lambda_u (\prod_{i=1}^{n}x_i^{u_i}) = \sum_{u \in \FF^n}\lambda_{u}X^u\,,
 $$
 with $\lambda_{u} \in \KK$ and $u=(u_1,\ldots,u_n)$.
\end{definition}
It can be proved 
that any Boolean function $f$ admits a unique numerical normal form.
As for the ANF, it is possible to compute the NNF of a Boolean function from its truth table by mean of an algorithm similar to a fast Fourier transform, thus requiring $O(n2^n)$ additions over $\KK$ and storing $O(2^n)$ elements of $\KK$.\\
\indent
From now on let $\KK = \QQ$.\\
The truth table of $f$ can be recovered from its NNF by the formula $$f(u)=\sum_{a\preceq u}\lambda_a,\forall u \in \FF^n\,,$$
where $a\preceq u\iff \forall i \in \{1,\ldots,n\} \; a_i \le u_i$. Conversely, 
it is possible to derive an explicit formula for the coefficients of the NNF by means of the truth table of $f$.
\begin{proposition}\label{propNNFcoeff}
 Let $f$ be any integer-valued function on $\FF^n$. For every $u\in \FF^n$, the coefficient $\lambda_u$ of the monomial $X^u$ in the NNF of $f$ is:
 \begin{equation}\label{eqNNFCoeff}
  \lambda_u = (-1)^{\w(u)}\sum_{a\in \FF^n |a\preceq u}(-1)^{\w(a)}f(a)\,.
 \end{equation}
\end{proposition}


\subsection{Nonlinearity of a Boolean function}

%
\begin{definition}
Let $f,g\in\B_n$. The distance $\dist(f,g)$ between $f$ and $g$ is the number of $v\in\FF^n$ such that $f(v)\neq g(v)$.
\end{definition}
The following lemma is obvious:
\begin{lemma}\label{distance}
Let $f,g$ be two Boolean functions. Then
$$\dist(f,g)=\dist(\ef,\eg)=\w(\ef+\eg)\,.$$
\end{lemma}
\begin{definition}
Let $f\in\B_n$. The \emph{nonlinearity} of $f$ is the minimum of the distances between $f$ and any affine function
$$\N(f)=\min_{\alpha\in\A_n}\dist(f,\alpha)\,.$$
\end{definition}
%
The maximum nonlinearity for a Boolean function $f$ is bounded by:
\begin{align}\label{eqMaxNL}
\max\{\N(f) \mid f\in\B_n\} \le 2^{n-1}-2^{\frac{n}{2}-1}\,.
\end{align}

\subsection{Walsh transform of a Boolean function}

\begin{definition}
The \emph{Walsh transform} of a Boolean function $f\in\B_n$ is the following function:
$$
\hat{F}: \FF^n \longrightarrow \mathbb{Z} 
\qquad
x \longmapsto \sum_{y\in\FF^n}(-1)^{x\cdot y + f(y)}\,.
$$
where $x\cdot y$ is the scalar product of $x$ and $y$.
\end{definition}
We have the following fact:
\begin{fact}
 $$\N(f)=\min_{v\in\FF^n}\{2^{n-1}-\frac{1}{2}\hat{F}(v)\}=2^{n-1}-\frac{1}{2}\max_{v\in\FF^n}\{\hat{F}(v)\}$$
\end{fact}
\begin{definition}
  The set of integers $\{\hat{F}(v) \mid v\in\FF^n\}$ is called the \emph{Walsh spectrum} of the Boolean function $f$.
\end{definition}
It is possible 
to compute the Walsh spectrum of $f$ from its evaluation vector in $O(n2^{n})$ integer operations, while storing $O(2^n)$ integers, by means of the \emph{fast Walsh transform} (the Walsh transform is the Fourier transform of the sign function of $f$).
Thus the computation of the nonlinearity of a Boolean function $f$, when this is given either in its ANF or in its evaluation vector, requires $O(n2^n)$ integer operations and a memory of $O(2^n)$.\\
\indent
Faster methods are known in particular cases, for example when the ANF is a sparse polynomial \cite{CGC-cry-phdthesis-calik2013}, \cite{CGC-cry-art-calik2013nonlinearity}.

  \section{Preliminary results}
  \label{secPrelOnGB}
%
%
  Here we present the main results from 
\cite{CGC-cd-art-ilawcc07}, \cite{CGC-cd-inbook-D1simonetti}. The same techniques are also applied in \cite{CGC-cd-prep-elemanumax} and \cite{CGC-tesi2-guerrini}.
\subsection{Polynomials and vector weights}
Let $\KK$ be a field and $X=\{x_1,\ldots,x_s\}$ be a set of variables. We denote by $\KK[X]$ the multivariate polynomial ring in the variables X. If $f_1,\ldots,f_N \in \KK[X]$, we denote by $\langle\{f_1,\ldots,f_N\}\rangle$ the ideal in $\KK[X]$ generated by $f_1,\ldots,f_N$. \\
Let $q$ be the power of a prime.
We denote by $E_q[X]=\{x_1^q-x_1,\ldots,x_s^q-x_s\}\,,$ the set of field equations in $\FF_q[X]=\FF_q[x_1,\ldots,x_s]$, where $s\geq 1$ is an integer, understood from now on. We write $E[X]$ when $q=2$.
\begin{definition} 
Let $1\leq t \leq s$ and ${\sf m}\in \FF_q [X]$. We say that ${\sf m}$ is 
a {\bf{square free monomial}} of degree $t$ (or a {\bf{simple $t$-monomial}}) if: 
$$
{\sf m} = x_{h_1}\cdots x_{h_t}, \textrm{ where } h_1,\ldots, h_t \in \{1,\ldots,s\} \textrm{ and } h_\ell \neq h_j, \forall \ell\neq j\, ,$$ 
i.e. a monomial in $\FF_q [X]$ such that $\deg_{x_{h_i}}({\sf m})=1$ for any $1\leq i \leq t$. We 
denote by $\mathcal{M}_{s,t}$  the set of all square free monomials of degree $t$ in $\FF_q [X]$.
\end{definition}
Let $t\in\NN$, with $1\leq t\leq s$ and let $I_{s,t}\subset\FF_q[X]$  be the following ideal
$$I_{s,t}=\langle\{\sigma_t,\ldots,\sigma_s\}\cup E_q[X]\rangle\,,$$
where $\sigma_i$ are the elementary symmetric functions:
$$
\begin{array}{lcl}
\sigma_1 & = & x_1+x_2+\cdots+x_s,\\
\sigma_2 & = & x_1x_2+x_1x_3+\cdots+x_1x_s+x_2x_3+\cdots+x_{s-1}x_s,\\
 & \cdots\\
 \sigma_{s-1} & = & x_1x_2x_3\cdots x_{s-2}x_{s-1}+\cdots+x_2x_3\cdots x_{s-1}y_s,\\
 \sigma_s & = & x_1x_2\cdots x_{s-1}x_s.
\end{array}
$$
We also denote by $I_{s,s+1}$ the ideal $\langle E_q[X] \rangle$.
For any $1\leq i\leq s$, let $P_i$ be the set which contains all vectors in $(\FF_q)^n$ of weight $i$, $P_i=\{v\in\FF_q^n\mid \w(v)=i\}$, and let $Q_i$ be the set which contains all vectors of weight up to $i$, $Q_i=\sqcup_{0\leq j\leq i}P_j$ .
\begin{theorem}\label{gpeso}
Let $t$ be an integer such that $1\leq t\leq s$. Then the vanishing ideal $\mathcal{I}(Q_t)$ of $Q_{t}$ is
$$\mathcal{I}(Q_t)=I_{s,t+1}\,,$$
and its reduced \Gr\ basis $G$ is
$$
\begin{array}{lcl}
G=E_q[X]\cup\mathcal{M}_{s,t}\,, & \quad & \textrm{for } t\geq 2\,,\\
G=\{x_1,\ldots,x_s\}\,,& \quad & \textrm{for } t=1\,. 
\end{array}
$$
\end{theorem}
%
%
Let $\FF_q[Z]$ be a polynomial ring over $\FF_q$. Let ${\sf m}\in\mathcal{M}_{s,t}$, ${\sf m}=z_{h_1}\cdots z_{h_t}$. For any polynomial vector $W$ in the module $(\FF_q[Z])^n$, $W=(W_1,\ldots,W_n)$, we denote by ${\sf m}(W)$ the following polynomial in $\FF_q[Z]$:
$${\sf m}(W)=W_{h_1}\cdot\ldots\cdot W_{h_t}\,.$$
\begin{example}
Let $n=s=3,q=2$ and $W=(x_1x_2+x_3,x_2,x_2x_3)\in(\FF[x_1,x_2,x_3])^3$ and ${\sf m}=z_1z_3$. Then
$${\sf m}(W)=(x_1x_2+x_3)(x_2x_3)\,.$$
\end{example}

\section{Computing the nonlinearity using \Gr\ bases over $\FF$}
  \label{secNLwithGBoverF2}
  In this section we show how to use Theorem \ref{gpeso} to compute the nonlinearity of a given Boolean function $f\in \B_n$.\\
  We want to define an ideal such that a point in its variety corresponds to an affine function with distance at most $t-1$ from $f$.\\
%

Let $A$ be the variable set $A=\{a_i\}_{0\leq i\leq n}$. We denote by $\mathfrak{g}_n\in\FF[A,X]$
the following polynomial:
$$\mathfrak{g}_n=a_0+\sum_{i=1}^n a_i x_i\,
\,.$$
%
%
\noindent According to Lemma \ref{distance}, determining the nonlinearity of $f\in\B_n$ is the same as finding the minimum weight of the vectors in the set $\{\ef+\eg\mid g\in\A_n\}\subset\FF^{2^n}$.
%
We can consider the evaluation vector of the polynomial $\mathfrak{g}_n$ as follows:
$$\underline{\mathfrak{g_n}}=(\mathfrak{g}_n(A,\op_1),\ldots,\mathfrak{g}_n(A,\op_{2^n}))\in (\FF[A])^{2^n}\,.$$
\begin{example}
Let $\mathfrak{g}_3$ be a general affine function in $\mathcal{A}_3$. Then $\mathfrak{g}_3=a_1x_1+a_2x_2+a_3x_3+a_0$. We consider vectors in $\FF^3$ ordered as follows:
$$
\begin{array}{cccc}
\op_1=(0,0,0), & \op_2=(0,0,1), & \op_3=(0,1,0), & \op_4=(1,0,0),\\
\op_5=(0,1,1), & \op_6=(1,0,1), & \op_7=(1,1,0), & \op_8=(1,1,1).
\end{array}
$$
So we have that the evaluation vector of $\mathfrak{g}_3$ is:
$$\underline{\mathfrak{g}_3}=(a_0,a_0+a_1,a_0+a_2,a_0+a_3,a_0+a_1+a_2,a_0+a_1+a_3,a_0+a_2+a_3,a_0+a_1+a_2+a_3)\,.$$
\end{example}
%
%
  %
%
%
\begin{definition}\label{defIdealF2}
We denote by $J_t^n(f)$ the ideal in $\FF[A]$:
$$
\begin{array}{ll}
J_t^n(f) 
& =
\langle
\{
{\sf m} \big (\mathfrak{g}_n(A,\op_1)+ f(\op_1),\ldots,\mathfrak{g}_n(A,\op_{2^n})+f(\op_{2^n})\big)\mid {\sf m}\in\mathcal{M}_{2^n,t}
\}
\cup 
E[A]
\rangle\\
& = 
\langle\{{\sf m}(\underline{\mathfrak{g}_n}+\ef)\mid {\sf m}\in\mathcal{M}_{2^n,t}\}\cup E[A]\rangle\,.
\end{array}
$$
\end{definition} 
\begin{remark}
As $E[A]\subset J_t^n(f)$, $J_t^n(f)$ is zero-dimensional and radical (\cite{CGC-alg-art-seidenberg1}).
\end{remark}
\begin{lemma}\label{nf}
For $1\leq t\leq 2^n$ the following statements are equivalent:
\begin{enumerate}
\item $\mathcal{V}(J_t^n(f))\neq \emptyset$,
\item $\exists u\in \{\ef+\eg\mid g\in\A_n\} \textrm{ such that } \w(u)\leq t-1$,
\item $\exists \alpha\in\A_n \textrm{ such that } \dist(f,\alpha)\leq t-1$.
\end{enumerate}
\end{lemma}
\begin{proof}~\\
(2)$\Leftrightarrow$(3). Obvious.\\
(1)$\Rightarrow$(2). 
Let $\bar{A}=(\bar{a}_0,\bar{a}_1,\ldots,\bar{a}_n)\in
     \mathcal{V}(J_t^n(f))\subset\FF^{n+1}$ 
     and let 
$u=(\mathfrak{g}_n(\bar{A},v_1)+f(v_1),\ldots,\mathfrak{g}_n(\bar{A},v_{2^n})+f(v_{2^n})) \in \FF^{2^n}$.
We have that ${\sf m}(u)=0$ for all ${\sf m}\in\mathcal{M}_{2^n,t}$.
So $u\in\mathcal{V}(I_{2^n,t})$ and, 
thanks to Theorem \ref{gpeso}, $u\in Q_{t-1}$, i.e. ${\rm w}(u)\leq t-1$.\\
(2)$\Rightarrow$(1). It can be proved by reversing the above argument.
\end{proof}
From Lemma \ref{nf} we immediately have the following theorem.
\begin{theorem}\label{Nf}
Let $f\in\B_n$. The nonlinearity $\N(f)$ is the minimum $t$ such that $\mathcal{V}(J_{t+1}^n(f))\neq \emptyset$.
\end{theorem}
From this theorem we can derive an algorithm to compute the nonlinearity for a function $f\in\B_n$, by computing any \Gr\ basis of $J_t^n(f)$.
\begin{algorithm}[H]
\caption{Basic algorithm to compute the nonlinearity of a Boolean function using \Gr\ basis over $\FF$}
\label{algNLoverF2}
  \begin{algorithmic}[1]
    \REQUIRE{a Boolean function $f$}
    \ENSURE{the nonlinearity of $f$}
    \STATE{$j \leftarrow 1$}
    \WHILE{${\mathcal V}(J_j^n(f)) =\emptyset $} 
      \STATE{$j \leftarrow j+1$} 
    \ENDWHILE
    \RETURN $j-1$
  \end{algorithmic}
\end{algorithm}
%
%
%
%
%
%
%
%
%
%
%
%
%
\begin{remark}
If $f$ is not affine, we can start our check from $J_2^n(f)$.
\end{remark}
\begin{example}
Let $f:\FF^3\rightarrow \FF$ be the Boolean function:
$$f(x_1,x_2,x_3)=x_1x_2+x_1x_3+x_2+1\,.$$
We want to compute $\N(f)$ and clearly $f$ is not affine. We compute vector $\ef$ and we take a general affine function $\mathfrak{g}_3$, so that:\\
$\ef=(1,1,0,1,1,0,0,0)$ ,\\
$\underline{\mathfrak{g}_3}=(a_0,a_0+a_1,a_0+a_2,a_0+a_3,a_0+a_1+a_2,a_0+a_1+a_3,a_0+a_2+a_3,a_0+a_1+a_2+a_3).$\\
So $\ef+\underline{\mathfrak{g}_3}=(a_0+1,a_0+a_1+1,a_0+a_2,a_0+a_3+1,a_0+a_1+a_2+1,a_0+a_1+a_3,a_0+a_2+a_3,a_0+a_1+a_2+a_3)=(p_1,p_2,\ldots,p_8)$ .\\
Ideal $J_2^3(f)$ is the ideal generated by
$$J_2^3(f)=\langle\{p_1p_2,p_1p_3,\ldots,p_7p_8\}\cup\{a_0^2+a_0,a_1^2+a_1,a_2^2+a_2,a_3^2+a_3\}\rangle\,.$$
We compute any \Gr\ basis of this ideal and we obtain that it is trivial, so ${\mathcal V}(J_2^3(f)) =\emptyset$ and $\N(f)>1$. Now we have to compute a \Gr\ basis for $J_3^3(f)$. We obtain, using degrevlex ordering with $a_1>a_2>a_3>a_0$, that $G(J_3^3(f))=\{ a_2+a_3+1,a_3^2+a_3,a_1a_3+a_0+1,a_0a_3+a_0+a_3+1, a_1^2+a_1, a_0a_1+a_0+a_1+1, a_0^2+a_0\}$. So, $\N(f)=2$ by Theorem \ref{Nf}. By inspecting $G(J_3^3(f))$, we also obtain all affine functions having distance 2 from $f$:
$$
\begin{array}{cccc}
\alpha_1=1+x_1+x_2, &\; \alpha_2=1+x_2, & \;\alpha_3=1+x_3, & \;\alpha_4=x_1+x_3\,.
\end{array}
$$
\end{example}
\begin{example}
Let $f:\FF^5\rightarrow \FF$ be the Boolean function
$$f=x_1x_3x_4x_5  + x_1x_2x_4 +  x_1x_4x_5 + x_2x_3x_4 + x_2x_4x_5 + x_3x_4x_5 + x_4x_5\,.$$
We have that 
$$\ef=(0,0,0,0,0,0,0,0,0,0,0,0,0,0,0,1,0,1,0,0,0,0,1,0,0,0,0,0,0,0,0,1)\,.$$ 
Then we compute $\ef+\underline{\mathfrak{g}_5}$ and we obtain:
$$
\begin{array}{rl}
\ef+\underline{\mathfrak{g}_5} = & (a_0,a_1+a_0,a_2+a_0,a_3+a_0,a_4+a_0,a_5+a_0,a_1+a_2+a_0,\\
& a_1+a_3+a_0,a_1+a_4+a_0,a_1+a_5+a_0,a_2+a_3+a_0, a_2+a_4+a_0,\\
& a_2+a_5+a_0,a_3+a_4+a_0,a_3+a_5+a_0,a_4+a_5+a_0+1,\\
& a_1+a_2+a_3+a_0, a_1+a_2+a_4+a_0+1,a_1+a_2+a_5+a_0,\\
& a_1+a_3+a_4+a_0,a_1+a_3+a_5+a_0, a_1+a_4+a_5+a_0,\\
& a_2+a_3+a_4+a_0+1,a_2+a_3+a_5+a_0,a_2+a_4+a_5+a_0,\\
& a_3+a_4+a_5+a_0,a_1+a_2+a_3+a_4+a_0,a_1+a_2+a_3+a_5+a_0,\\
& a_1+a_2+a_4+a_5+a_0,a_1+a_3+a_4+a_5+a_0,a_2+a_3+a_4+a_5+a_0,\\
& a_1+a_2+a_3+a_4+a_5+a_0+1)= (p_1,p_2,\ldots,p_{32})\,.
\end{array}
$$
As it is obvious that $f$ is not affine, we start from the ideal $J_2^5(f)$, which is generated by
$$J_2^5(f)=\langle\{p_1p_2,p_1p_3,\ldots,p_{31}p_{32}\}\cup\{a_0^2+a_0,a_1^2+a_1,a_2^2+a_2,a_3^2+a_3,a_4^2+a_4,a_5^2+a_5\}\rangle\,.$$
The \Gr\ basis of $J_2^5(f)$ with respect to any monomial order is trivial so we compute a \Gr\ basis of $J_3^5(f)$. We obtain that the \Gr\ basis of $J_t^5(f)$ is trivial with respect to any monomial order for $2\leq t\leq 4$. For $t=5$, we obtain the following \Gr\ basis with respect to the degrevlex order with $a_1>a_2>a_3>a_4>a_5>a_0$:
$$G(J_5^5(f))=\{a_0,a_5,a_4,a_3,a_2,a_1\}\,.$$
Then $\N(f)=4$, that is, there is only one affine function $\alpha$ which has distance equal to $4$ from $f$: $\alpha=0$.
\end{example}
%


\section{Computing the nonlinearity using \Gr\ bases over $\QQ$}
  \label{secNLwithGBoverQ}
%
%
Here we present an algorithm to compute the nonlinearity of a Boolean function using \Gr\ bases over $\QQ$ rather than over $\FF$, which turns out to be
much faster than Algorithm \ref{algNLoverF2}. The same algorithm can be slightly modified to work over the field $\FF_p$, where $p$ is a prime. The complexity of these algorithms will be analyzed in Section \ref{secNLComplexity}.\\
\indent
As we have seen in Section \ref{secNLwithGBoverF2}, the nonlinearity of a Boolean function can be computed using \Gr\ bases over $\FF$. It is sufficient to find the minimum $j$ such that the variety of the ideal $J_t^n(f)$ is not empty. Recall that
$$
J_t^n(f) = \langle\{{\sf m}(\underline{\mathfrak{g}_n}+\ef)\mid {\sf m}\in\mathcal{M}_{2^n,t}\}\cup E[A]\rangle\,.
$$
This method becomes impractical even for small values of $n$, since $\binom{2^n}{t}$ monomials have to be evaluated. A first slight improvement could be achieved by adding to the ideal one monomial evaluation at a time and check if 1 has appeared in the \Gr\ basis. Even this way, the algorithm remains very slow.\\
%
%
%
For each $i=1,\ldots,2^n$, let us denote:
$$f_{i}^{(\FF)}(A)=\mathfrak{g}_n(A,\op_i)+f(\op_i)$$
the Boolean function where as usual $A = \{a_0,\dots,a_n\}$ are the $n+1$ variables representing the coefficient of a generic affine function.\\
In this case we have that:
$$(f_1^{(\FF)}(A),\dots,f_{2^n}^{(\FF)}(A)) = \underline{\mathfrak{g}_n}(A)+\ef \in (\FF[A])^{2^n}$$
Note that the polynomials $f_{i}^{(\FF)}$ are affine polynomials. 
\\
We also denote by
$$f_i^{(\ZZ)}(A) = \text{NNF}(f_{i}^{(\FF)}(A))$$
the NNF of each $f_{i}^{(\FF)}(A)$ (obtained as in \cite{CGC-cry-art-carlet1999}, Theorem 1).
\begin{definition}\label{defNLP}
 We call $\nP_f(A) = f_1^{(\ZZ)}(A)+\dots+f_n^{(\ZZ)}(A) \in \ZZ[A]$ the {\bf integer nonlinearity polynomial} (or simply the \emph{nonlinearity polynomial}) of the Boolean function $f$.\\
 For any $t\in \NN$ we define the ideal $\NLI_f^t \subseteq \QQ[A]$ as follows:
 \begin{align}
  \NLI_f^t =
  \langle E[A] \bigcup \{ f_1^{(\ZZ)}+\dots+f_{2^n}^{(\ZZ)}-t \} \rangle =
  \langle E[A] \bigcup \{ \nP_f-t \} \rangle
 \end{align}
\end{definition}
 Note that the evaluation vector $\underline{\nP_f}$ represents all the distances of $f$ from all possible affine functions (in $n$ variables).
%
\begin{theorem}
The variety of the ideal $\NLI_f^t$ is non-empty if and only if
the Boolean function $f$ has distance $t$ from an affine function. 
In particular, $\N(f) = t$, where $t$ is the minimum positive integer such that $\mathcal{V}(\NLI_f^t)\ne \emptyset$.
\end{theorem}
\begin{proof}
 Note that 
 $$\NLI_f^t = \langle E[A] \rangle + \langle \{ \nP_f(A)-t \} \rangle$$
 and so 
 $$\mathcal{V}(\NLI_f^t) = \mathcal{V}(\langle E[A]\rangle) \cap \mathcal{V}(\langle \{ \nP_f(A)-t \} \rangle)\,.$$
 Therefore $\mathcal{V}(\NLI_f^t) \ne \emptyset$ if and only if 
 $\exists \bar{a}=(\bar{a}_0, \ldots, \bar{a}_n) \in \mathcal{V}(\langle E[A]\rangle)$ such that $\nP_f(\bar{a})=t$.\\
 Let $\alpha \in \mathcal{A}_n$ such that $\alpha(X) = \bar{a}_0 + \sum_{i=1}^n \bar{a}_ix_i$.\\
 By definition we have 
 $$f_i^{(\ZZ)} = 1 \iff f(\op_i) \ne \alpha(\op_i)$$ 
 and 
 $$f_i^{(\ZZ)} = 0 \iff f(\op_i) = \alpha(\op_i)\,.$$
 Hence 
 $$\nP_f(\bar{a}) = \sum_{i=1}^{2^n}f_i^{(\ZZ)}(\bar{a})-t = 0 \iff |\{i \mid f(\op_i)\ne \alpha(\op_i) \}|=t \iff \dd(f,\alpha) = t\,.$$
 and our claim follows directly.
\end{proof}
To compute the nonlinearity of $f$ we can use Algorithm \ref{algNLoverQ} with input $f$.\\
%
\begin{algorithm}[H]
\caption{To compute the nonlinearity of the Boolean function $f$}
\label{algNLoverQ}
\begin{algorithmic}[1]
\REQUIRE{$f$}
\ENSURE{nonlinearity of $f$}
\STATE{Compute $\nP_f$}
\STATE{$j \leftarrow 1$}
\WHILE{$\V(\NLI_f^j) = \emptyset$}
  \STATE{$j \leftarrow j+1$}
\ENDWHILE
\RETURN{j} 
\end{algorithmic}
\end{algorithm}
%


\section{Computing the nonlinearity using fast polynomial evaluation}
  \label{secNLwithFPE}
%

Once the nonlinearity polynomial $\nP_f$ is defined, we can use another approach to compute the nonlinearity avoiding the computations of \Gr\ bases.\\ 
We have to find the minimum nonnegative integer $t$ in the set of the evaluations of $\nP_f$, that is, in $\{\nP_f(\bar{a}) \mid \bar{a} \in \{0,1\}^{n+1} \subset \ZZ^{n+1}\}$.\\
We write explicitly the modified algorithm.

\begin{algorithm}[H]
\caption{To compute the nonlinearity of the Boolean function $f$}
\label{algNLfromNLP}
\begin{algorithmic}[1]
\REQUIRE{$f$}
\ENSURE{nonlinearity of $f$}
\IF{$f \in \mathcal{A}_n$}
  \RETURN{$0$}
\ELSE
  \STATE{Compute $\nP_f$}
  \STATE{Compute $m = \min\{\nP_f(\bar{a}) \mid \bar{a} \in \{0,1\}^{n+1} \}$}
  \RETURN{$m$}
\ENDIF
\end{algorithmic}
\end{algorithm}

\begin{example}
 Consider the case $n=2$, $f(x_1,x_2) = x_1x_2 + 1$. We have that $\ef = (1,1,1,0)$ and $\underline{\mathfrak{g}_n}=(a_0,a_0+a_1,a_0+a_2,a_0+a_1+a_2)$.\\ 
 Let us compute all $f_i^{(\FF)}=(\underline{\mathfrak{g}_n}+\ef)_i$ and $f_i^{(\ZZ)}$,for $i=1,\ldots,2^2$:
 \begin{align*}
  f_1^{(\FF)} & = a_0 + 1          & \rightarrow f_1^{(\ZZ)} &= -a_0 + 1\\  
  f_2^{(\FF)} & = a_0 + a_1 + 1    & \rightarrow f_2^{(\ZZ)} &= 2a_0a_1 - a_0 - a_1 + 1\\  
  f_3^{(\FF)} & = a_0 + a_2 + 1    & \rightarrow f_3^{(\ZZ)} &= 2a_0a_2 - a_0 - a_2 + 1\\  
  f_4^{(\FF)} & = a_0 + a_1  + a_2 & \rightarrow f_4^{(\ZZ)} &= 4a_0a_1a_2 - 2a_0a_1 - 2a_0a_2 + a_0 - 2a_1a_2 + a_1 + a_2
 \end{align*}
 Then $\nP_f = f_1^{(\ZZ)} + f_2^{(\ZZ)} + f_3^{(\ZZ)} + f_4^{(\ZZ)} = 4a_0a_1a_2 - 2a_0 - 2a_1a_2 + 3$ and since
 $$\underline{\nP_f} = (3,1,3,1,3,1,1,3)$$
 then the nonlinearity of $f$ is $1$.\\
 Observe that the vector $\underline{\nP_f}$ represents all the distances of $f$ from all possible affine functions in $2$ variables, that is, from $0,1,x_1,x_1+1,x_2,x_2+1,x_1+x_2,x_1+x_2+1$.
\end{example}
%
  

\section{Properties of the nonlinearity polynomial}
  \label{secNLPolProp}
%
%
%
From now on, with abuse of notation, we sometimes consider $0$ and $1$ as elements of $\FF$ and other times as elements of $\ZZ$.\\
We have the following definition
\begin{definition}\label{defBinToInt}
 Given $b_1,\ldots,b_n \in \FF$
$$b_1 \oplus \ldots \oplus b_n = 
\sum_{{\bold v}=(v_1,\ldots,v_n)\in \FF^n,{\bold v} \ne {\bold 0}}
(-2)^{\w({\bold v})-1}\cdot
b_1^{v_1}\cdots b_n^{v_n}\,.$$
where the sum on the right is in $\ZZ$.
\end{definition}
It is easy to show that $b_1 \oplus \ldots \oplus b_n \in \{0,1\}$.\\
We give a theorem to compute the coefficients of the nonlinearity polynomial.
\begin{theorem}\label{thmNPCoeffFormula}
  Let 
  $v=(v_0,v_1,\ldots,v_n)\in\FF^{n+1}$,
  $\tilde{v}=(v_1,\ldots,v_n)\in\FF^{n}$, 
  $A^v=a_0^{v_0}\cdots a_n^{v_n} \in \FF[A]$ 
  and 
  $c_v \in \ZZ$ be such that
  $\nP_f = \sum_{v\in\FF^{n+1}}c_vA^v$. 
  Then the coefficients of $\nP_f$ can be computed as:
  \begin{align}  
  \label{eqCoeff0}
  c_{v} = \sum_{u\in \FF^{n}}f(u) = \w(\underline{f}) \,\,\text{ if } v = 0
  \\
  \label{eqCoeffv}
  c_v =
  (-2)^{\w(v)}
  \sum_{\substack{u\in \FF^n \\ \tilde{v}\preceq u}} 
  \left[
  f(u) - \frac{1}{2}
  \right]
  \,\,\text{ if } v \ne 0
  \end{align}
\end{theorem}
\begin{proof}
 The nonlinearity polynomial is the integer sum of the $2^n$ numerical normal forms of the affine polynomials $\mathfrak{g}_n(A,u)\oplus f(u)\in \FF[A]$, each identified by the vector $u\in \FF^n$, i.e.:
 \begin{align*}
  \nP_f = 
  \sum_{u\in \FF^n}\text{NNF}(\mathfrak{g}_n(A,u)\oplus f(u)) = 
  \sum_{u\in \FF^n}\text{NNF}(a_0 \oplus a_1u_1 \oplus \ldots \oplus a_nu_n \oplus f(u))
 \end{align*}
 which is a polynomial in $\ZZ[A]$.\\
 The NNF of $\mathfrak{g}_n(A,u)\oplus f(u)$ is a polynomial with $2^{n+1}$ terms, i.e.:
 \begin{align*}
  \text{NNF}(\mathfrak{g}_n(A,u)\oplus f(u)) = \sum_{v \in \FF^{n+1}}\lambda_{v}A^v\,,
 \end{align*}
 for some $\lambda_v \in \ZZ$, and by Proposition \ref{propNNFcoeff}
 \begin{align*}
  \lambda_{v}(u) = (-1)^{\w(v)}\sum_{a\in \FF^{n+1} |a\preceq v}(-1)^{\w(a)}\Big(\mathfrak{g}_n(a,u)\oplus f(u)\Big)\,.
 \end{align*}
 Let us prove Equation (\ref{eqCoeff0}). When $v=(0,\ldots,0)$ we have
 \begin{align*}
  c_{(0,\ldots,0)} 
  & = 
  \sum_{u\in \FF^n} 
  \big[\mathfrak{g}_n((0,\ldots,0),u)\oplus f(u)
  \big]
  = 
  \sum_{u\in \FF^n} f(u)\,.
 \end{align*}
 Let us prove Equation (\ref{eqCoeffv}). Suppose $v\ne0$.\\
 Now the coefficient $c_v$ of the monomial $A^v$ of the nonlinearity polynomial is such that:
 \begin{align}
  c_v & = \sum_{u\in \FF^n}\lambda_{v}(u) = \notag \\
  & = \sum_{u\in \FF^n} (-1)^{\w(v)}\sum_{\substack{a\in \FF^{n+1},\\ 
  a\preceq v}}(-1)^{\w(a)}\big[\mathfrak{g}_n(a,u)\oplus f(u)\big] = \notag \\
  \label{eqCoeff_v1}
  & = (-1)^{\w(v)} \sum_{u\in \FF^n} \sum_{\substack{a\in \FF^{n+1},\\ a\preceq v}}(-1)^{\w(a)}\big[\mathfrak{g}_n(a,u)\oplus f(u)\big] \,. 
\end{align}
We prove that each $u$ such that $\tilde{v}=(v_1,\ldots,v_n)\npreceq u$ yields a zero term in the summation, as follows.\\
If $\tilde{v}\npreceq u$ then $\exists i \in \{1,\ldots,n\}$ s.t. $v_i>u_i$, i.e. $v_i = 1,u_i=0$. We claim that
$\forall a \in \FF^{n+1}$ s.t. $a \preceq v \,\,\,
 \exists \bar{a}=(\bar{a}_0,\ldots,\bar{a}_n) \in \FF^{n+1}$ s.t. $\bar{a} \preceq v$ and
\begin{align}\label{eqaabar}
 (-1)^{\w(a)}\big[\mathfrak{g}_n(a,u)\oplus f(u)\big] +
 (-1)^{\w(\bar{a})}\big[\mathfrak{g}_n(\bar{a},u)\oplus f(u)\big] = 0
\end{align}
It is sufficient to choose $\bar{a}_i \ne a_i$ and $\bar{a}_j = a_j$ for all $j \in \{1,\ldots,n\}, j\ne i$. Clearly $\bar{a}\preceq v$ and $a\preceq v$ since $v_i = 1$.\\
By direct substitution we obtain
\begin{align*}
   & (-1)^{\w(a)}\big[\mathfrak{g}_n(a,u)\oplus f(u)\big] +
 (-1)^{\w(\bar{a})}\big[\mathfrak{g}_n(\bar{a},u)\oplus f(u)\big] = \\
 = &
 (-1)^{\w(a)}\big[a_0 \oplus a_1u_1 \oplus \ldots \oplus a_iu_i \oplus \ldots \oplus a_nu_n \big] + \\
 & (-1)^{\w(a)}(-1)\big[\bar{a}_0 \oplus \bar{a}_1u_1 \oplus \ldots \oplus \bar{a}_iu_i \oplus \ldots \oplus \bar{a}_nu_n \big] \\
 = & (-1)^{\w(a)}[a_iu_i - \bar{a}_iu_i] = 0\,.
\end{align*}
Thanks to (\ref{eqaabar}) we can continue from (\ref{eqCoeff_v1}) and get
\begin{align}
  \label{eqCoeff_v2}
  c_v & = 
  (-1)^{\w(v)} 
  \sum_{\substack{u\in \FF^n \\ \tilde{v}\preceq u}} 
  \sum_{\substack{a\in \FF^{n+1},\\ a\preceq v}}(-1)^{\w(a)}\big[\mathfrak{g}_n(a,u)+ f(u)-2\mathfrak{g}_n(a,u)f(u)\big] 
  \,,
 \end{align}
 where we used $a\oplus b = a+b-2ab$.\\
 \indent
 Now we consider $v,u$ fixed, and $\tilde{v}\preceq u$. \\
 There are exactly $2^{\w(v)}$ vectors $a$ such that $a\preceq v$, i.e.: 
 \begin{align}
 |\{ a \in \FF^{n+1} \mid a\preceq v \}| = 2^{\w(v)}
 \end{align}
 Now we want to study the internal summation in (\ref{eqCoeff_v2}).\\
 If $u=(0,\ldots,0)$ then $\forall a=(a_0,\ldots,a_n) \preceq v$ we have $\mathfrak{g}_n(a,u) = a_0\oplus a_1u_1\oplus\ldots a_nu_n = a_0$.\\
 Otherwise, if $u\ne(0,\ldots,0)$ we can consider the following set of indices $U=\{j \mid u_j = 1\} = \{j_1,\ldots,j_{\w(u)}\}$, which has size $\w(u)$.\\
 Since $a \preceq v$ and $\tilde{v}\preceq u$ then $(a_1,\ldots,a_n) \preceq u$ by transitivity. For all $j \notin U$ we have $a_j = 0$, and then $\w(a_0, a_{j_1},\ldots,a_{j_{\w(u)}}) = \w(a)$.\\
 Thus, for any $u \in \FF^n$ we have 
 \begin{align}
  \mathfrak{g}_n(a,u) = a_0 \oplus a_{j_1}\oplus \ldots\oplus a_{j_{\w(u)}} = 
  \begin{cases}
   1 \mbox{ if } \w(a) \mbox{ is odd} \\ 
   0 \mbox{ if } \w(a) \mbox{ is even}
  \end{cases}
 \end{align}
 and each of the two cases happens for exactly one half of the vectors $a \preceq v$. Clearly the two halves are disjoint.\\
 This yields, from (\ref{eqCoeff_v1}) and (\ref{eqCoeff_v2}), the following chain of equalities:
 \begin{align*}
  c_v & = \sum_{u\in \FF^n}\lambda_{v}(u) = \notag 
  \\ 
  & = (-1)^{\w(v)} 
  \sum_{\substack{u\in \FF^n \\ \tilde{v}\preceq u}} 
  \bigg[ 
  \sum_{\substack{a\in \FF^{n+1},\\ 
        a\preceq v\\
        \mathfrak{g}_n(a,u)=0}}
  (-1)^{\w(a)}f(u)
  +
  \sum_{\substack{a\in \FF^{n+1},\\ 
        a\preceq v\\
        \mathfrak{g}_n(a,u)=1}}
  (-1)^{\w(a)}(1-f(u))
  \bigg]
  =
 \\
  & = (-1)^{\w(v)} 
  \sum_{\substack{u\in \FF^n \\ \tilde{v}\preceq u}} 
  \bigg[ 
  \sum_{\substack{a\in \FF^{n+1},\\ 
        a\preceq v\\
        \mathfrak{g}_n(a,u)=0}}
  f(u)
  +
  \sum_{\substack{a\in \FF^{n+1},\\ 
        a\preceq v\\
        \mathfrak{g}_n(a,u)=1}}
  (f(u)-1)
  \bigg]
  =
 \\
  & = (-1)^{\w(v)} 
  \sum_{\substack{u\in \FF^n \\ \tilde{v}\preceq u}} 
  \bigg[
  2^{\w(v)-1}f(u) +
  2^{\w(v)-1}(f(u)-1) 
  \bigg]
  =
 \\ 
   & = (-1)^{\w(v)} 
  \sum_{\substack{u\in \FF^n \\ \tilde{v}\preceq u}} 
  \bigg[
  2^{\w(v)}f(u) 
  - 2^{\w(v)-1}
  \bigg]  
  =
 \\ 
   & = (-2)^{\w(v)} 
  \sum_{\substack{u\in \FF^n \\ \tilde{v}\preceq u}} 
  \bigg[  
  f(u) - \frac{1}{2}
  \bigg]  
\end{align*} 
which proves the theorem.
\end{proof}
In particular we have:
\begin{corollary}\label{corNPHalfCoeff}
  Let $u=(u_1,\ldots,u_n)$ and $\nP_f = 
  \sum_{u\in \FF^{n}}c_{(0,u)}a_1^{u_1}\cdot\ldots\cdot a_n^{u_n} + a_0
  \sum_{u\in \FF^{n}}c_{(1,u)}a_1^{u_1}\cdot\ldots\cdot a_n^{u_n}$. Then we have that:
 \begin{align}\label{eqNPc1}
  c_{(1,0,\ldots,0)} 
  = 2^n-2\w(\underline{f})
 \end{align}
 And $\forall \tilde{v} \in \FF^{n}, \tilde{v} \ne 0$ we have:
 \begin{align}\label{eqNP1c}
  c_{(1,\tilde{v})} = -2 c_{(0,\tilde{v})},\,.
 \end{align}
\end{corollary}
Corollary \ref{corNPHalfCoeff} shows that it is sufficient to store half of the coefficients of $\nP_f$, precisely the coefficients of the monomials where $a_0$ does not appear.
%
%
%
\begin{corollary}\label{corNPCoeffSize}
 Each coefficient $c$ of the nonlinearity polynomial $\nP_f$ is such that $|c| \le 2^n$.
\end{corollary}
%
\begin{corollary}\label{corNPDual}
 Given the nonlinearity polynomial of $f$ as
 $$\nP_f(a_0,\ldots,a_n)=
 c_{(0,\ldots,0)} +
 \sum_{\underset{(p_0,\ldots,p_n) \ne (0,\ldots,0)}{(p_0,\ldots,p_n)\in \FF^{n+1}}}
 c_{(p_0,\ldots,p_n)}a_0^{p_0}\cdot\ldots\cdot a_n^{p_n}
 $$
 then the nonlinearity polynomial of $f\oplus1$ is related to that of $f$ by the following rule:
 $$\nP_{f\oplus1}(a_0,\ldots,a_n)=
 2^n-c_{(0,\ldots,0)} +
 \sum_{\underset{(p_0,\ldots,p_n) \ne (0,\ldots,0)}{(p_0,\ldots,p_n)\in \FF^{n+1}}}
  -c_{(p_0,\ldots,p_n)}a_0^{p_0}\cdot\ldots\cdot a_n^{p_n}
 $$
\end{corollary}
%
A scheme that shows how to derive the coefficients of the nonlinearity polynomial in the case $n=3$ can be seen in Tables \ref{tabNLPcoeff1} and \ref{tabNLPcoeff2}.
\begin{table}[h]
\begin{center}
\resizebox{14cm}{!} {
\begin{tabular}{ll| *8{l}|}%
 $u$& $f(u)+\mathfrak{g}_n(a_0,a_1,a_2,a_3,u)$ & $1$ & $a_3$ & $a_2$ & $a_2a_3$ & $a_1$ & $a_1a_3$ & $a_1a_2$ & $a_1a_2a_3$ \\
\hline 
000 & $v_1+a_0$             & $v_1$ &          &          &           &          &           &           &          \\
001 & $v_2+a_0+a_3$         & $v_2$ & $1-2v_2$ &          &           &          &           &           &          \\
010 & $v_2+a_0+a_2$         & $v_3$ &          & $1-2v_3$ &           &          &           &           &          \\
011 & $v_2+a_0+a_2+a_3$     & $v_4$ & $1-2v_4$ & $1-2v_4$ & $-2+4v_4$ &          &           &           &          \\
100 & $v_2+a_0+a_1$         & $v_5$ &          &          &           & $1-2v_5$ &           &           &          \\
101 & $v_2+a_0+a_1+a_3$     & $v_6$ & $1-2v_6$ &          &           & $1-2v_6$ & $-2+4v_6$ &           &          \\
110 & $v_2+a_0+a_1+a_2$     & $v_7$ &          & $1-2v_7$ &           & $1-2v_7$ &           & $-2+4v_7$ &          \\
111 & $v_2+a_0+a_1+a_2+a_3$ & $v_8$ & $1-2v_8$ & $1-2v_8$ & $-2+4v_8$ & $1-2v_8$ & $-2+4v_8$ & $-2+4v_8$ & $4-8v_8$ \\
\end{tabular}
}
\end{center}
\caption{Computation of the coefficients of the nonlinearity polynomial with $n=3$. Each line represents the NNF coefficients of the terms of $f(u)+\mathfrak{g}_n(A,u)$ \emph{not} containing $a_0$.}
\label{tabNLPcoeff1}
\end{table}
\begin{table}[h]
\begin{center}
\resizebox{14cm}{!} {
\begin{tabular}{ll| *8{l}|}%
 $u$& $f(u)+\mathfrak{g}_n(a_0,a_1,a_2,a_3,u)$ & $a_0$ & $a_0a_3$ & $a_0a_2$ & $a_0a_2a_3$ & $a_0a_1$ & $a_0a_1a_3$ & $a_0a_1a_2$ & $a_0a_1a_2a_3$ \\
\hline 
000 & $v_1+a_0$             & $1-2v_1$ &           &           &          &           &          &          &            \\
001 & $v_2+a_0+a_3$         & $1-2v_2$ & $-2+4v_2$ &           &          &           &          &          &            \\
010 & $v_2+a_0+a_2$         & $1-2v_3$ &           & $-2+4v_3$ &          &           &          &          &            \\
011 & $v_2+a_0+a_2+a_3$     & $1-2v_4$ & $-2+4v_4$ & $-2+4v_4$ & $4-8v_4$ &           &          &          &            \\
100 & $v_2+a_0+a_1$         & $1-2v_5$ &           &           &          & $-2+4v_5$ &          &          &            \\
101 & $v_2+a_0+a_1+a_3$     & $1-2v_6$ & $-2+4v_6$ &           &          & $-2+4v_6$ & $4-8v_6$ &          &            \\
110 & $v_2+a_0+a_1+a_2$     & $1-2v_7$ &           & $-2+4v_7$ &          & $-2+4v_7$ &          & $4-8v_7$ &            \\
111 & $v_2+a_0+a_1+a_2+a_3$ & $1-2v_8$ & $-2+4v_8$ & $-2+4v_8$ & $4-8v_8$ & $-2+4v_8$ & $4-8v_8$ & $4-8v_8$ & $-8+16v_8$ \\
\end{tabular}
}
\end{center}
\caption{Computation of the coefficients of the nonlinearity polynomial with $n=3$. Each line represents the NNF coefficients of the terms of $f(u)+\mathfrak{g}_n(A,u)$ containing $a_0$.}
\label{tabNLPcoeff2}
\end{table}
%


\section{Complexity of constructing the nonlinearity polynomial}
  \label{secNLPolComplex}
%
%
We write the algorithm (Algorithm \ref{algNLPcoeff}) to calculate the nonlinearity polynomial in $O(n2^n)$ integer operations.
\begin{algorithm}[H]
\caption{Algorithm to calculate the nonlinearity polynomial  $\nP_f$ in $O(n2^n)$ integter operations.}
\label{algNLPcoeff}
  \begin{algorithmic}[1]
    \REQUIRE{The evaluation vector $\underline{f}$ of a Boolean function $f(x_1,\ldots,x_n)$}
    \ENSURE{the vector $c=(c_1,\ldots,c_{2^{n+1}})$ of the coefficients of $\nP_f$} \\
    Calculation of the coefficients of the monomials not containing $a_0$
    \STATE{$(c_1,\ldots,c_{2^n}) = \underline{f}$}
    \FOR{$i=0,\ldots,n-1$}
      \STATE{$b\leftarrow 0$}
      \REPEAT
        \FOR{$x = b,\ldots,b+2^i-1$}
        \STATE{$c_{x+1} \leftarrow c_{x+1} + c_{x+2^i+1}$}\label{stepSUM} 
        \IF{$x = b$}
          \STATE{$c_{x+2^i+1} \leftarrow 2^i -2c_{x+2^i+1}$}\label{stepG1} 
        \ELSE
          \STATE{$c_{x+2^i+1} \leftarrow -2c_{x+2^i+1}$}\label{stepG2} 
        \ENDIF
        \ENDFOR
        \STATE{$b \leftarrow b+2^{i+1}$} 
      \UNTIL{$b=2^n$}
    \ENDFOR \\
  Calculation of the coefficients of the monomials containing $a_0$
  \STATE{$c_{1+2^n} \leftarrow 2^n -2c_{1}$}
  \FOR{$i = 2,\ldots,2^n$}
    \STATE{$c_{i+2^n} \leftarrow-2c_{i}$}
  \ENDFOR
    \RETURN $c$
  \end{algorithmic}
\end{algorithm}
In Figure \ref{figNLPn3} Algorithm \ref{algNLPcoeff} is shown for $n=3$.
\begin{figure}[h]
\scalebox{.6}{
 \xymatrix@=10pt{
  (x_1,x_2,x_3) & f(x_1,x_2,x_3)           & &   & Step\;1     &                         & &   & Step\;2                 &                           & &   & Step\;3                                          & \\  
  000        & e_{1}\ar[rr]                & & + & e_{1}+e_{2} & \ar[rr]                 & & + & e_{1}+e_{2}+e_{3}+e_{4} & \ar[rr]                   & & + & e_{1}+e_{2}+e_{3}+e_{4}+e_{5}+e_{6}+e_{7}+e_{8}  & \\
  001        & e_{2}\ar[urr]\ar[rr]_{1-2x} & &   & 1-2e_{2}    & \ar[rr]                 & & + & 2-2e_{2}-2e{4}          & \ar[rr]                   & & + & 4-2e_{2}-2e_{4}-2e_{6}-2e_{8}                    & \\
  010        & e_{3}\ar[rr]                & & + & e_{3}+e_{4} & \ar[uurr]\ar[rr]_{2-2x} & &   & 2-2e_{3}-2e{4}          & \ar[rr]                   & & + & 4-2e_{3}-2e_{4}-2e_{7}-2e_{8}                    & \\
  011        & e_{4}\ar[urr]\ar[rr]_{1-2x} & &   & 1-2e_{4}    & \ar[uurr]\ar[rr]_{-2x}  & &   & -2+4e_{4}               & \ar[rr]                   & & + & -4+4e_{4}-4e_{8}                                 & \\
  100        & e_{5}\ar[rr]                & & + & e_{5}+e_{6} & \ar[rr]                 & & + & e_{5}+e_{6}+e_{7}+e_{8} & \ar[uuuurr]\ar[rr]_{4-2x} & &   & 4-2e_{5}-2e_{6}-2e_{7}-2e_{8}                    & \\
  101        & e_{6}\ar[urr]\ar[rr]_{1-2x} & &   & 1-2e_{6}    & \ar[rr]                 & & + & 2-2e_{6}-2e{8}          & \ar[uuuurr]\ar[rr]_{-2x}  & &   & -4+4e_{6}-4e_{8}                                 & \\
  110        & e_{7}\ar[rr]                & & + & e_{7}+e_{8} & \ar[uurr]\ar[rr]_{2-2x} & &   & 2-2e_{7}-2e{8}          & \ar[uuuurr]\ar[rr]_{-2x}  & &   & -4+4e_{7}-4e_{8}                                 & \\
  111        & e_{8}\ar[urr]\ar[rr]_{1-2x} & &   & 1-2e_{8}    & \ar[uurr]\ar[rr]_{-2x}  & &   & -2+4e_{8}               & \ar[uuuurr]\ar[rr]_{-2x}  & &   & 4-8e_{8}                                         & \\
 }
}
\caption{Butterfly scheme to obtain a fast computation of the nonlinearity polynomial coefficients, where $(e_1,\ldots,e_8)=(f(\op_1),\ldots,f(\op_8))$.}
\label{figNLPn3}
\end{figure}
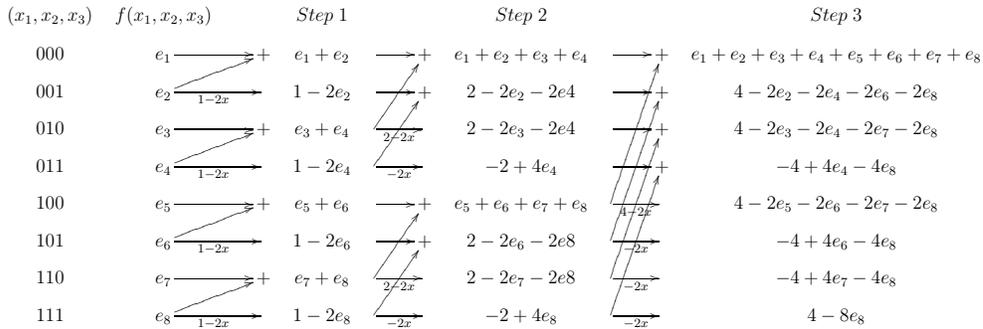
\begin{theorem}\label{thmNLPn2n}
 Algorithm \ref{algNLPcoeff} requires:
 \begin{enumerate}
  \item $O(n2^n)$ integer sums and doublings, in particular circa $n2^{n-1}$ integer sums and circa $n2^{n-1}$ integer doublings.
  \item the storage of $O(2^n)$ integers of size less than or equal to $2^n$.
 \end{enumerate}
\end{theorem}
\begin{proof}
 In the first part of Algorithm \ref{algNLPcoeff} (the computation of the coefficients of the monomials not containing $a_0$) the iteration on $i$ is repeated $n$ times.\\
 For each $i$, Step \ref{stepSUM} and Step \ref{stepG1} or \ref{stepG2} are repeated $2^i\frac{2^n}{2^{i+1}}=2^{n/2}$ times (since $b$ goes from $0$ to $2^n$ by a step of $2^{i+1}$ and $x$ performs $2^i$ steps). In Step \ref{stepSUM} only one integer sum is performed, in Steps \ref{stepG1} we have one integer sum and one doubling, and in Step \ref{stepG2} only one doubling. Then the total amount of integer operation is 
 $$O(n2^n)$$
 Finally the computation of the coefficients of the monomials containing $a_0$ requires only $2^n$ integer doublings.\\
 To store all the monomials of the nonlinearity polynomial we have to store $2^{n+1}$ integers, although Corollary \ref{corNPHalfCoeff} shows that it is sufficient to store only the first half of them, i.e. $2^n$ integers. By Corollary \ref{corNPCoeffSize}, their size is less than or equal to $2^n$.
\end{proof}
%
  

\section{Complexity considerations}
  \label{secNLComplexity}
%
First we recall that the complexity of computing the nonlinearity of a Boolean function with $n$ variables, having as input its coefficients vector, is $O(n2^n)$ using the Fast M\"obius and the Fast Walsh Transform.\\
We now want to analyze the complexity of Algorithm \ref{algNLoverF2}, \ref{algNLoverQ}, \ref{algNLfromNLP}.
%
%
%
\subsection{Some considerations on Algorithm \ref{algNLoverF2}}
In Algorithm \ref{algNLoverF2}, almost all the computations are wasted evaluating all possible simple-$t$-monomials in $2^n$ variables, which are $\binom{2^n}{t}$. This number grows enormously even for small values of $n$ and $t$. We investigated experimentally how many of the $\binom{2^n}{t}$ monomials are actually needed to compute the final \Gr\ basis of $J_t^n$. Our experiment ran over all possible Boolean functions in 3 and 4 variables. The results are reported in Tables \ref{tabNumberOfMonomialsJt3},
\ref{tabNumberOfMonomialsJt4t1_3} and \ref{tabNumberOfMonomialsJt4t4_7}.\\
In this tables, for each $J_t^n$ there are four columns. Let $G_t^n$ be the \Gr\ basis of $J_t^n$. \\
Under the column labeled $\#$C we report the average number of \emph{checked} monomials in $2^n$ variables before obtaining $G_t^n$. \\
Under the column labeled $\#$S we report the average number of monomials which are actually \emph{sufficient} to obtain $G_t^n$. \\
Under the columns labeled ``m'' e ``M'' we report, respectively, the minimum and the maximum number of sufficient monomials to find $G_t^n$ running through all possible Boolean functions in $n$ variables.
\\
For example, to compute the \Gr\ basis of the ideal $J_2^3$ associated to a Boolean function $f$ whose nonlinearity is $2$, we needed to check on average 24 monomials before finding the correct basis. Between the $24$ monomials only $9.7$ (on average) were sufficient to obtain the same basis, where the number of sufficient monomials never exceeded the range $8-11$.
\begin{table}[h]
\begin{center}
\resizebox{14cm}{!} {
\begin{tabular}{c|*3{*4{c}|}}
    & \multicolumn{4}{c}{$J_1^3$} &   \multicolumn{4}{c}{$J_2^3$} &   \multicolumn{4}{c}{$J_3^3$} \\
 NL & $\#$S & m & M & $\#$C &     $\#$S & m & M & $\#$C &     $\#$S & m & M & $\#$C \\
\hline
  0 &   4   & 4 & 4 & 8    & 0   & 0 &  0 & 0     & 0   & 0 & 0  & 0  \\
  1 &   4.5 & 4 & 5 & 4.4  & 8.5 & 7 & 10 & 28    & 0   & 0 & 0  & 0  \\
  2 &   4.4 & 4 & 5 & 4    & 9.7 & 8 & 11 & 24    & 9.3 & 8 & 11 & 56
%
\end{tabular}
}
\end{center}
\caption{Number of monomials needed to compute the \Gr\ basis of the ideal $J_t^3$.}
\label{tabNumberOfMonomialsJt3}
\end{table}
\begin{table}[h]
\begin{center}
\resizebox{14cm}{!} {
\begin{tabular}{c|*3{*4{l}|}}
 & \multicolumn{4}{c}{$J_1^4$} &   \multicolumn{4}{c}{$J_2^4$} &   \multicolumn{4}{c}{$J_3^4$} \\
NL & $\#$S & m & M & $\#$C &     $\#$S & m & M & $\#$C &     $\#$S & m & M & $\#$C \\
\hline
  0 &   5    & 5 & 5 & 16     &  0    & 0 & 0  & 0       & 0 & 0 & 0 & 0              \\  
  1 &   5.25 & 4 & 6 & 8      &  8.75 & 8 & 11 & 120     & 0 & 0 & 0 & 0              \\
  2 &   4.83 & 4 & 6 & 5.67   &  9.97 & 8 & 12 & 62.83   & 14.50 & 12 & 18 & 560      \\
  3 &   4.62 & 4 & 6 & 4.76   &  9.92 & 8 & 12 & 42.72   & 15.76 & 13 & 19 & 315.04   \\
  4 &   4.53 & 4 & 6 & 4.42   &  9.83 & 8 & 12 & 37.49   & 15.81 & 13 & 19 & 246.19   \\
  5 &   4.46 & 4 & 5 & 4.19   & 10.11 & 8 & 12 & 34.39   & 15.89 & 13 & 19 & 215.68   \\
  6 &   4.43 & 4 & 5 & 4.00   &  9.71 & 8 & 11 & 24.00   & 17.29 & 16 & 19 & 156.86   
\end{tabular}
}
\end{center}
\caption{Number of monomials needed to compute the \Gr\ basis of the ideal $J_t^4$, $t=1,2,3$.}
\label{tabNumberOfMonomialsJt4t1_3}
\end{table}
\begin{table}[h]
\begin{center}
\resizebox{14cm}{!} {
\begin{tabular}{c|*4{*4{l}|}}
 & \multicolumn{4}{c}{$J_4^4$} &   \multicolumn{4}{c}{$J_5^4$} &   \multicolumn{4}{c}{$J_6^4$} &   \multicolumn{4}{c}{$J_7^4$}    \\
NL & $\#$S & m & M & $\#$C &     $\#$S & m & M & $\#$C &     $\#$S & m & M & $\#$C &     $\#$S & m & M & $\#$C \\
\hline
  0   & 0 & 0 & 0 & 0               & 0 & 0 & 0 & 0               & 0 & 0 & 0 & 0               & 0 & 0 & 0 & 0  \\  
  1   & 0 & 0 & 0 & 0               & 0 & 0 & 0 & 0               & 0 & 0 & 0 & 0               & 0 & 0 & 0 & 0  \\
  2   & 0 & 0 & 0 & 0               & 0 & 0 & 0 & 0               & 0 & 0 & 0 & 0               & 0 & 0 & 0 & 0  \\
  3   & 20.18 & 15 & 23 & 1820      & 0 & 0 & 0 & 0               & 0 & 0 & 0 & 0               & 0 & 0 & 0 & 0  \\
  4   & 21.44 & 16 & 24 & 1319.96   & 23.99 & 22 & 29 & 4368      & 0 & 0 & 0 & 0               & 0 & 0 & 0 & 0  \\
  5   & 21.54 & 19 & 24 & 1003.15   & 26.00 & 24 & 28 & 3851.24   & 23.50 & 22 & 25 & 8008      & 0 & 0 & 0 & 0  \\
  6   & 19.57 & 19 & 20 &  671.71   & 28    & 28 & 28 & 2603.79   & 28 & 28 & 28 & 7608.79      & 16 & 16 & 16 & 11441
\end{tabular}
}
\end{center}
\caption{Number of monomials needed to compute the \Gr\ basis of the ideal $J_t^4$,$t=4,5,6,7$.}
\label{tabNumberOfMonomialsJt4t4_7}
\end{table}
\subsection{Algorithm \ref{algNLoverF2} and \ref{algNLoverQ}}
Since it is not easy to estimate the complexity of a \Gr\ basis computation theoretically, we give some experimental results, shown in Table \ref{tabNLtimings}. 
\begin{table}[h]
\begin{center}
\resizebox{14cm}{!} {
\begin{tabular}{c|cccccc}
$n$   & $\log_2\big[\frac{(n+1)2^{n+1}}{n2^{n}}\big]$ & FWT & NLP+FPE & GB on $\FF_p$ & GB on $\QQ$ & GB on $\FF$\\
\hline
2-3   & 1.53 &   -  &   -  & 1.45 & 1.86 & 2.50 \\
3-4   & 1.31 &   -  &   -  & 1.88 & 2.27 & 7.51 \\ 
4-5   & 1.22 & 0.90 & 1.02 & 2.33 & 2.91 &  -   \\
5-6   & 1.17 & 0.98 & 1.09 & 2.64 & 3.23 &  -   \\
6-7   & 1.14 & 1.01 & 1.13 & 2.76 & 4.29 &  -   \\
7-8   & 1.12 & 1.22 & 1.07 & 3.24 &  -   &  -   \\
8-9   & 1.11 & 0.95 & 1.17 & 3.48 &  -   &  -   \\
9-10  & 1.09 & 1.25 & 1.07 &  -   &  -   &  -   \\
10-11 & 1.09 & 1.07 & 1.11 &  -   &  -   &  -
\end{tabular}
}
\end{center}
\caption{Experimental comparisons of the coefficients of growth of the analyzed algorithms.}
\label{tabNLtimings}
\end{table}
In this table we report the coefficients of growth of the analyzed algorithms
\footnote{To compute the values in the columns FWT and NLP+FPE we tested $15000$ random Boolean functions from $n=4$, since for $n=3$ there are only $2^{(2^3)}=256$ Boolean functions.}
, comparing them with the value $\log_2\big[\frac{(n+1)2^{n+1}}{n2^{n}}\big]$. For each algorithm we compute the average time $t_n$ to compute the nonlinearity of a Boolean function with $n$ variables and the average time $t_{n+1}$ to compute the nonlinearity of a Boolean function with $n+1$ variables. Then we report in the table the value $\log_2\big(\frac{t_{n+1}}{t_{n}}\big)$. When \Gr\ bases are computed, then graded reverse lexicographical order is used, with Magma \cite{CGC-MAGMA} implementation of the Faugère $F4$ algorithm. 
Since the ideal $J_t^n(f)$ of Definition \ref{defIdealF2} is derived from the evaluation of $\binom{2^n}{t}$ monomials (generating at most the same number of equations), then the complexity of Algorithm \ref{algNLoverF2} is equivalent to the complexity of computing a \Gr\ basis of at most $\binom{2^n}{t}$ equations of degree $d$ (where $1 < d \le t$) in $n+1$ variables over the field $\FF$. This method becomes almost impractical for $n=5$.
We recall that $t\le 2^{n-1}-2^{\frac{n}{2}-1}$ (see Equation \ref{eqMaxNL}).\\
\indent
The complexity of Algorithm \ref{algNLoverQ} is equivalent to the complexity of computing a \Gr\ basis of only  $n+1$ field equations plus one single polynomial $\nP_f$ of degree at most $n+1$ in $n+1$ variables over the field $\QQ$ (or over a prime field $\FF_p$) with coefficients of size less then or equal to $2^n$. \\
As shown in Table \ref{tabNLtimings}, computing this \Gr\ basis over a prime field $\FF_p$ with $p\sim 2^n$ is much faster than computing the same base over $\QQ$. It may be investigated if there are better size for the prime $p$.\\
%
%
\subsection{Algorithm \ref{algNLfromNLP}}
\begin{theorem}
 Algorithm \ref{algNLfromNLP} returns the nonlinearity of a Boolean function $f$ with $n$ variables in
 $$O(n2^{n})$$
 integers operations (sums and doublings).
\end{theorem}
\begin{proof}
 Algorithm \ref{algNLfromNLP} can be divided in three main steps:
 \begin{enumerate}
  \item Calculation of the nonlinearity polynomial $\nP_f$. This step, as shown in Theorem \ref{thmNLPn2n}, requires $O(n2^n)$ integer operations and $O(2^n)$ memory.
  \item Evaluation of the nonlinearity polynomial $\nP_f$. This step can be performed using fast M\"obius transform in $O(n2^n)$ integer sums and $O(2^n)$ memory.
  \item Computation of the minimum $\nP_f(a)$ with $a\in \ZZ^{n+1}$. This step requires no more than $O(2^n)$ checks.
 \end{enumerate}
 The overall complexity is then $O(n2^n)$ integer operations and $O(2^n)$ memory.
\end{proof}


\section{Conclusions}
  \label{secConcl}

We presented an approach to compute the nonlinearity of a Boolean
function using multivariate polynomials. 
In particular we show that the problem of computing the distance of a generic Boolean function $f$ from the set of affine functions is equivalent to the problem of solving a multivariate polynomial system  over the binary field. 
This system can be reformulated over the rationals by considering the associated pseudo Boolean function, and we can exhibit a multivariate polynomial whose evaluations solve the problem. Moreover, we evaluate our polynomial using fast Fourier techniques and solve the problem very efficienlty. In particular, with our polynomial-based approach we compute the nonlinearity of any Boolean function in $O(n2^n)$  operations, reaching the same complexity of classical methods.
  

\section{Acknowledgments}
  \label{secAck}
 \noindent
The first two authors would like to thank the third author (their supervisor).


\bibliography{RefsCGC}




\end{document}